\documentclass[a4paper,reqno,11pt]{amsart}

\usepackage{latexsym,amssymb,amsmath,amsthm,amsfonts,enumerate,verbatim,xspace,
exscale}

\input xy
\xyoption{all} \CompileMatrices \UseComputerModernTips

\usepackage[all]{xy} 
\usepackage{amssymb} 
\usepackage{hyperref}
\usepackage{eucal}
\usepackage{graphicx}
\usepackage{xcolor}

\usepackage{graphicx,tabularx}
\usepackage{amsmath,amsbsy,amsfonts,amssymb}
\usepackage{color}
\usepackage{amsthm}

\title{A global version of the Koon-Marsden Jacobiator formula}

\author[Paula \ Balseiro]{Paula \ Balseiro}
\address{Paula \ Balseiro:
Universidade Federal Fluminense, Instituto de Matem\'atica, Rua Mario Santos Braga S/N, 24020-140,  Niteroi, Rio de Janeiro, Brazil.} \email{pbalseiro@vm.uff.br}

\theoremstyle{plain}
\newtheorem{theorem}{Theorem}[section]
\newtheorem{lemma}[theorem]{Lemma}

\newtheorem*{theorem*}{Theorem}
\newtheorem{remarkth}[theorem]{Remark}

\theoremstyle{definition}

\newenvironment{remark}{\begin{remarkth}\upshape}{\hfill$\diamond$\end{remarkth}}

\def\W{\mathcal{W}}
\def\M{\mathcal{M}}

\def\V{\mathcal{V}}
\def\S{\mathcal{S}}
\def\C{\mathcal{C}}
\def\Ham{\mathcal{H}}

\def\R{\mathbb{R}}

\def\L{\mbox{Leg}}
\def\red{{\mbox{\tiny{red}}}}
\def\nh{{\mbox{\tiny{nh}}}}

\def\B{{\mbox{\tiny{$B$}}}}
\def\subW{{\mbox{\tiny{$\W$}}}}
\def\subC{{\mbox{\tiny{$\C$}}}}
\def\subS{{\mbox{\tiny{$\S$}}}}
\def\subM{{\mbox{\tiny{$\M$}}}}

\def\a{\alpha}

\begin{document}
\maketitle

\begin{abstract}
In this paper we study the Jacobiator (the cyclic sum that vanishes
when the Jacobi identity holds) of the almost Poisson brackets
describing nonholonomic systems. We revisit the local formula for
the Jacobiator established by Koon and Marsden in \cite{MarsdenKoon}
using suitable local coordinates and explain how it is related to
the global formula obtained in \cite{paula}, based on the choice of
a complement to the constraint distribution. We use an example to
illustrate the benefits of the coordinate-free viewpoint.
\end{abstract}

\begin{center} {\it Dedicated to the memory of J.E. Marsden}
\end{center}

\tableofcontents

\section{Introduction} \label{S:Intro}
The geometric approach to nonholonomic systems was among the many
research interests of J. E. Marsden, and his contributions to this
area were fundamental. A system with nonholonomic constraints can be
geometrically described by an almost Poisson bracket
\cite{IbLeMaMa1999,Marle1998,SchaftMaschke1994}, whose failure to
satisfy the Jacobi identity, measured by the so-called Jacobiator,
is precisely what encodes the nonholonomic nature of the system.
There is a vast literature on the study of such nonholonomic
brackets and their properties, starting with the early work of
Chaplygin \cite{Chapligyn_reducing_multiplier}, see e.g.
\cite{BS93,IbLeMaMa1999,MovingFrames,Fernandez,BorisovMamaev2008,Naranjo2008,JovaChap}. An explicit formula
for the Jacobiator of nonholonomic brackets, expressed in suitable
local coordinates, was obtained by Koon and Marsden in their 1998
paper \cite{MarsdenKoon}. In the present paper, we revisit the
Koon-Marsden formula of \cite{MarsdenKoon} and explain how it can be
derived from the coordinate-free Jacobiator formula for nonholonomic
brackets obtained in \cite{paula}.

We organize the paper as follows. In Section~\ref{S:NHsystem}, we
recall the hamiltonian viewpoint to systems with nonholonomic
constraints. For a nonholonomic system on a configuration manifold
$Q$, determined by a lagrangian $L:TQ \to \R$ and a nonintegrable
distribution $D$ on $Q$ (the {\it constraint distribution}, defining
the permitted velocities of the system), we consider the induced
{\it nonholonomic bracket}  $\{ \cdot, \cdot \}_\nh$ defined on the
submanifold $\M:=\L(D)$ of $T^*Q$, where $\L:T^*Q \to TQ$ is the
Legendre transform (see Section~\ref{Sub:nhb}). In
Section~\ref{Sub:jac} (see Theorem~\ref{T:MarsdenKoon}) we recall
the global formula for the Jacobiator of $\{\cdot,\cdot\}_\nh$ from
\cite{paula}, which depends on the choice of a complement $W$ of the
constraint distribution $D$ such that $TQ=D \oplus W$. As shown in
\cite{paula}, this formula is useful to provide information about
properties of reduced nonholonomic brackets in the presence of
symmetries.

In Section~\ref{S:adapted} we recall the choice of coordinates,
suitably adapted to the constraints, used by Koon and Marsden in
\cite{MarsdenKoon}, and in terms of which their Jacobiator formula
is expressed. We then compare the global and local viewpoints in
Section~\ref{S:coord}, explaining how one can derive the local
Jacobiator formula in \cite{MarsdenKoon} from the coordinate-free
formula in \cite{paula}.

Since the formula in \cite{paula} is coordinate free, it can be used
in examples without specific choices of coordinates. We illustrate
this fact studying the {\it snakeboard}, following
\cite{Ostrowski,MarsdenKoon}; here the natural coordinates in the
problem are not adapted to the constraints so, in principle, the
local formula from \cite{MarsdenKoon} cannot be directly applied.

\medskip

\noindent {\bf Acknowledgments}: I thank the organizers of the {\it Focus
Program on Geometry, Mechanics and Dynamics, the Legacy of Jerry
Marsden}, held at the Fields Institute in Canada, for their
hospitality during my stay. I also benefited from the financial
support given by Mitacs (Canada), and I am specially grateful to
Jair Koiller for his help. I also thank FAPERJ (Brazil) and the GMC
Network (projects MTM2012-34478, Spain) for their support.

\section{Nonholonomic systems} \label{S:NHsystem}

\subsection{The hamiltonian viewpoint} \label{Sub:ham}

A nonholonomic system is a mechanical system on a configuration
manifold $Q$ with constraints on the velocities which are not
derived from constraints in the positions. Mathematically, it is
defined by a lagrangian $L :TQ \to \R$ of mechanical type, i.e., $L=
\kappa - U$ where $\kappa$ is the kinetic energy metric and $U \in
C^\infty(Q)$ is the potential energy, and a nonintegrable
distribution $D$ on $Q$ determining the constraints, see
\cite{BlochBook,CushmannBook}. 
If $D$ is an
integrable distribution then the system is called {\em holonomic}.

In order to have an intrinsic formulation of the dynamics of
nonholonomic systems,  let us consider the Legendre transform $\L:TQ
\to T^*Q$ associated to the lagrangian $L$. The Legendre transform
is a diffeomorphism since $\L = \kappa^\flat$, where
$\kappa^\flat:TQ \to T^*Q$ is defined by
$\kappa^\flat(X)(Y)=\kappa(X,Y)$. We denote by
$\Ham:T^*Q \to \R$ the hamiltonian function associated to the lagragian $L$. 

We define the constraint submanifold $\M$ of $T^*Q$ by $\M =
\kappa^\flat(D)$.  Note that $\M$ is a vector subbundle of $T^*Q$.
We denote by $\tau: \M \to Q$ the restriction to $\M$ of the
canonical projection $\tau_Q:T^*Q \to Q$.

On $\M$ we have a natural 2-form $\Omega_\subM$ given by
$\Omega_\subM := \iota^*\Omega_Q$  where $\iota :\M \to T^*Q$ is the
inclusion and $\Omega_Q$ is the canonical 2-form on $T^*Q$. The
constraints are encoded on a (regular) distribution $\C$ on $\M$
defined, at each $m \in \M$, by
\begin{equation} \label{Eq:C}
\C_m = \{ v \in T_m\M \ : \ T\tau(v) \in D_{\tau(m)} \}.
\end{equation}

It was proven in \cite{BS93} that the point-wise restriction of the
2-form $\Omega_\subM$ to $\C$, denoted by $\Omega_\subM|_\C$, is
nondegenerate. That is, if $X \in \Gamma(\C)$ is such that ${\bf
i}_X \Omega_\subM|_\C \equiv 0$, then $X =0$.   Therefore, there is
a unique vector field $X_\nh$ on $\M$, called the {\it nonholonomic
vector field}, such that $X_\nh(m) \in \C_m$ and
\begin{equation} \label{Eq:NH-Dyn}
{\bf i}_{X_\nh} \Omega_\M |_\C  = d\Ham_\subM|_\C,
\end{equation}
where $\Ham_\subM := \iota^*\Ham: \M \to \R$.
The integral curves of $X_\nh$ are solutions of the nonholonomic dynamics \cite{BS93}.

In order to write \eqref{Eq:NH-Dyn} in local coordinates, suppose
that the constraint distribution $D$ is described (locally) by the
annihilators of 1-forms $\epsilon^a$ for $a=1,...,k$, that is $D=
\{(q, \dot q) \ : \ \epsilon^a(q) (\dot q) = 0 \mbox{ for all }
a=1,...,k \}.$ If we consider canonical coordinates $(q^i, p_i)$ on
$T^*Q$ then the constraints are given by
$$
\epsilon^a_i (q) \frac{\partial \Ham}{\partial p_i} = 0, \qquad \mbox{for } a=1,...,k,
$$
and \eqref{Eq:NH-Dyn} becomes
$$
\dot q^i = \frac{\partial \Ham}{\partial p_i}, \qquad \dot p_i = - \frac{\partial \Ham}{\partial q^i} + \lambda_a \epsilon^a,
$$
where $\lambda_a$ are functions (called the Lagrange multipliers) which are uniquely determined by the fact that the constraints are satisfied.

\subsection{The nonholonomic bracket}\label{Sub:nhb}

Recall that an {\it almost Poisson bracket} on $\M$ is an
$\R$-bilinear bracket $\{\cdot, \cdot \}: C^\infty(\M) \times
C^\infty(\M) \to C^\infty(\M)$ that is skew-symmetric and satisfies
the Leibniz condition:
$$
\{fg,h\} = f\{g,h\} + \{f,h\}g, \qquad \mbox{for } f,g,h \in C^\infty(\M).
$$
If $\{\cdot, \cdot \}$ satisfies the Jacobi identity, then the
bracket is called {\it Poisson}. The {\it hamiltonian vector field}
$X_f$ on $\M$ associated to a $f \in C^\infty(\M)$ is defined by
\begin{equation} \label{Eq:HamVF}
X_f = \{ \cdot , f\}
\end{equation}
and the {\it characteristic distribution} of $\{\cdot, \cdot \}$ is
a distribution on the manifold $\M$ whose fibers are spanned by the
hamiltonian vector fields. If the bracket is Poisson, then its
characteristic distribution is integrable. However, the converse is
not always true.

From the Leibniz identity it follows that there is a one-to-one
correspondence between almost Poisson brackets $\{\cdot, \cdot \}$
and bivector fields $\pi \in \bigwedge^2(T\M)$ given by
\begin{equation} \label{Eq:Pi-bracket}
\{f,g\} = \pi(df,dg), \qquad f,g, \in C^\infty (\M).
\end{equation}
Let us denote by $\pi^\sharp:T^*\M \to T\M$ the map defined by
$\beta(\pi^\sharp(\alpha)) = \pi(\alpha, \beta).$ Then, using
\eqref{Eq:HamVF}, the hamiltonian vector field $X_f$ is also given
by $X_f = - \pi^\sharp(df)$ and the characteristic distribution of
$\pi$ is the image of $\pi^\sharp$. The Schouten bracket $[\pi,\pi]$
(see \cite{MarsdenRatiu}) measures the failure of the Jacobi
identity of $\{\cdot , \cdot \}$ through the relation
\begin{equation}
\frac{1}{2}[\pi, \pi](df,dg,dh) = \{f,\{g,h\}\} + \{f, \{g,h\}\}+ \{g,\{h,f\}\} + \{h,\{f,g\}\} \label{E:Jacobi}
\end{equation}
for $f, g, h \in C^\infty(\M)$. So we refer to the trivector
$\frac{1}{2}[\pi,\pi]$ as the {\it Jacobiator} of $\pi$, which is
zero when $\pi$ is a Poisson bivector.

Coming back to our context, consider a nonholonomic system
on a manifold $Q$ defined by a lagrangian $L$ and a constraint distribution $D$.   Due to the nondegeneracy  of $\Omega_\M|_\C$,
there  is an induced bivector field $\pi_\nh \in \bigwedge^2(T\M)$
defined at each $\alpha \in T^*\M$ by
\begin{equation} \label{Eq:Pinh}
\pi^\sharp_\nh(\a)=X \quad \mbox{if and only if} \quad {\bf i}_X \Omega_\subM |_\C = -\alpha|_\C.
\end{equation}
The characteristic distribution of $\pi_\nh$ is the distribution
$\C$ defined in \eqref{Eq:C}. Since $\C$ is not integrable,
$\pi_\nh$ is not Poisson.

The bivector field $\pi_\nh$ is called the {\it nonholonomic
bivector field} \cite{SchaftMaschke1994,Marle1998, IbLeMaMa1999} and
it describes the dynamics in the sense that
\begin{equation} \label{Eq:Xnh}
\pi_\nh^\sharp(d\Ham_\subM) = -X_\nh.
\end{equation}
By \eqref{Eq:Pi-bracket}, the nonholonomic bivector $\pi_\nh$
defines uniquely an almost Poisson bracket $\{\cdot, \cdot \}_\nh$
on $\M$, called the {\it nonholonomic bracket}. From \eqref{Eq:Pinh}
we observe that
$$
\{f,g \}_\nh = \Omega_\subM(X_f, X_g) \qquad \mbox{for } f,g \in C^\infty(\M),
$$
where $X_f=-\pi_\nh^\sharp(df)$ and $X_g=-\pi_\nh^\sharp(dg)$. The
nonholonomic vector field \eqref{Eq:Xnh} is equivalently defined
through the equation $X_\nh = \{ \cdot , \Ham_\subM\}_\nh$.

\subsection{The Jacobiator formula}\label{Sub:jac}

Recall that $\C$ is a smooth distribution on $\M$. Choose a
complement $\W$ of $\C$ on $T\M$ such that, for each $m \in \M$,
\begin{equation} \label{Eq:SplittingOfTM}
T_m\M = \C_m \oplus \mathcal{W}_m.
\end{equation}

Let $P_\subC :T\M \to \C$ and $P_\subW :T\M \to \W$ be the
projections associated to the decomposition
\eqref{Eq:SplittingOfTM}. Since $P_\subW :T\M \to \W$ can be seen as
a $\W$-valued 1-form, following \cite{paula}, we define the
$\W$-valued 2-form ${\bf K}_\subW$ given by
\begin{equation}
\label{Def:K} {\bf K}_\subW(X,Y) = - P_\subW( [P_\subC (X),
P_\subC(Y)] ) \qquad \mbox{for } X,Y \in \mathfrak{X}(\M).
\end{equation}

Once a complement $\W$ of $\C$ is chosen, we obtain a
coordinate-free formula for the Jacobiator of the nonholonomic
bracket.

\begin{theorem}\cite{paula}  \label{T:MarsdenKoon}
The following holds:
\begin{equation} \label{Eq:Jacobiator}
\frac{1}{2}[\pi_{\emph \nh}, \pi_{\emph \nh}] (\alpha, \beta, \gamma)=  \Omega_\M ({\bf K}_\subW (\pi^\sharp_{\emph \nh}(\alpha), \pi^\sharp_{\emph \nh}(\beta) ), \pi^\sharp_{\emph \nh}(\gamma)) - \gamma \left( {\bf K}_\subW(\pi^\sharp_{\emph \nh}(\alpha), \pi^\sharp_{\emph \nh}(\beta)) \right) + \textup{cyclic}.
\end{equation}
for $\alpha, \beta, \gamma \in T^*\M$.
\end{theorem}

In fact, a more general formula appeared in \cite{paula}, valid for
any bivector field $\pi_\B$ gauge related to $\pi_\nh$. In that
context, this formula was used to understand under which
circumstances the reduction of $\pi_\B$ by symmetries had an
integrable characteristic distribution (even if it was not Poisson).

We will now show how this formula recovers the coordinate Jacobiator
formula obtained in \cite{MarsdenKoon}.

\section{The Koon-Marsden adapted coordinates}\label{S:adapted}

In this section we will recall the Koon-Marsden approach to writing
the Jacobiator of a nonholonomic bracket, based on a suitable choice
of coordinates of the manifold $Q$. After this, we will write the
objects presented in Section \ref{S:NHsystem} (such as the 2-forms
$\Omega_\M$ and ${\bf K}_\subW$, and the bivector $\pi_\nh$) in such
local coordinates in order to see the equivalence between the local
and global viewpoints.

We start by recalling the coordinates chosen in \cite{MarsdenKoon}.
Consider a nonholonomic system given by a lagrangian $L$ and a
nonintegrable distribution $D$. Let $\epsilon^a$ for $a=1,...,k$ be
1-forms that span the annihilator of $D$, i.e., $D^\circ  =
\textup{span}\{\epsilon^a\}$. The authors in \cite{MarsdenKoon}
introduce local coordinates $(q^i) = (r^\a,s^a)$ on $Q$ for which
each 1-form $\epsilon^a$ has the form
\begin{equation}\label{Eq:CoordMK}
\epsilon^a = ds^a +A_\a^a(r,s) dr^\a,
\end{equation}
where $A_\a^a$ are functions on $Q$ for $\a=1,...,n-k$ and
$a=1,...,k$. During the present paper, we refer to the coordinates
$(r^\a, s^a)$ such that \eqref{Eq:CoordMK} is satisfied as {\it
coordinates adapted to the constraints}.

These coordinates induce a (local) basis of $D$ given by
$\left\{X_\a:= \frac{\partial}{\partial r^\a} - A_\a^a
\frac{\partial}{\partial s^a}\right\}$.  We complete the basis
$\{X_\a\}$ and $\{\epsilon^a\}$ in order to obtain dual basis on
$TQ$ and $T^*Q$, that is
$$
TQ=\textup{span}\left\{X_\a, \frac{\partial}{\partial s^a} \right\}
\quad \mbox{and} \quad T^*Q=\textup{span}\{dr^\a, \epsilon^a\}.
$$

Let $(\tilde p_\a, \tilde p_a)$ be the coordinates on $T^*Q$
associated to the basis $\{dr^\a, \epsilon^a\}$. Since $\M =
\textup{span}\{\kappa^\flat(X_\a)\} \subset T^*Q$ then
\begin{equation} \label{Eq:M}
\M=\{(q^i, \tilde p_a,\tilde p_\a) \ : \ \tilde p_a =
[\kappa_{a\a}][\kappa_{\a\beta}]^{-1}\tilde p_\beta =J_a^\beta
\tilde p_\beta\},
\end{equation}
where $[\kappa_{a\a}]$ denotes the $(k\times(n-k))$-matrix with
entries given by $\kappa_{a\a} = \kappa( \frac{\partial}{\partial
s^a} ,X_\a)$, $[\kappa_{\a\beta}]^{-1}$ is the inverse matrix
associated to the invertible $((n-k)\times(n-k))$-matrix with
entries given by $\kappa_{\a\beta} = \kappa(X_\a,X_\beta)$ and
$J_a^\beta$ are the functions on $Q$ representing the entries of the
matrix $[\kappa_{a\a}][\kappa_{\a\beta}]^{-1}$. Therefore, each
element $(r^\a, s^a; \tilde p_\a)$ represents a point on the
manifold $\M$.

In \cite{MarsdenKoon} the Jacobiator formula is written in terms of
the curvature of an Ehresmann connection. 
The local coordinates $(r^\a, s^a)$ induce a fiber bundle with projection  given by $\upsilon(r^\a,s^a) = r^\a$. 
Let us call $W$ the vertical distribution defined by this projection.

The Ehresmann connection $A$ on $\upsilon:Q=\{r^\a,s^a\}\to R=\{r^\a\}$ is chosen in such a way that its horizontal space agrees with the distribution $D$. The connection $A$ is represented by a
vector-valued differential form given, at each $X\in TQ$, by
\begin{equation} \label{Eq:A}
A(X)=\epsilon^a(X)\frac{\partial}{\partial s^a}.
\end{equation}

 The {\it curvature} associated to this connection is
a vector-valued 2-form ${\bf K}_W$ defined on $X,Y \in
\mathfrak{X}(Q)$ by
\begin{equation}\label{Eq:Kw-MK}
{\bf K}_W(X,Y) = -A([P_D (X), P_D(Y)] ),
\end{equation}
where $P_D:TQ \to TQ$ is the projection to $D$ given by $P_D(X)=dr^\a(X)X_\a$.

In coordinates, the curvature ${\bf K}_W$ is given by the following
formula \cite[Sec. 2.1]{MarsdenKoon}:
$$
{\bf K}_W(X,Y)= d \epsilon^a (P_D (X), P_D(Y))\frac{\partial}{\partial s^a},
$$
hence, locally,
\begin{equation} \label{Eq:depsilon-coord}
d\epsilon^a |_D = C_{\a\beta}^a dr^\a \wedge dr^\beta |_D,
\end{equation}
where $\displaystyle{C_{\a\beta}^a (r,s)=  \frac{\partial
A_\beta^a}{\partial r^\a} - A_\a^b  \frac{\partial
A_\beta^a}{\partial s^b} }$. Let us define
\begin{equation} \label{Eq:Calphabeta}
K_{\a\beta}^a = C_{\a\beta}^a - C_{\beta\a}^a.
\end{equation}
For each $a=1,...,k$ the coefficients $K_{\a\beta}^a$ are
skew-symmetric and $d\epsilon^a |_D = K_{\a\beta}^a dr^\a \wedge
dr^\beta |_D,$ for $\a < \beta$. Therefore, if $X,\bar{X} \in D$
then $d\epsilon^a (X, \bar{X}) = K_{\a\beta}^a v^\a\bar{v}^\beta$
where $X= v^\a X_\a$ and $\bar{X}=\bar{v}^\beta X_\beta$.

\begin{remark}
Observe that in \cite{MarsdenKoon}, the 1-forms $\epsilon^a$ where denoted by $\omega^a$ while ${\bf K}_W$ was denoted by $B$ and the coefficients $K_{\a\beta}^a$ were $-B_{\a\beta}^a$. In this case, for  $\dot q \in D$ then $d\omega^b(\dot q, \cdot )|_D = -B_{\a\beta}^b \dot r^\a dr^\beta |_D$ (observe the correction in the sign with respect to the equation in \cite[Sec.~2.1]{MarsdenKoon}).
\end{remark}

Finally, in  \cite[Theorem 2.1]{MarsdenKoon} the almost Poisson
bracket $\{ \cdot, \cdot \}_\subM$ describing the dynamics of a
nonholonomic system was written following \cite{SchaftMaschke1994}
but in local coordinates on $Q$ adapted to the constraints
\eqref{Eq:CoordMK}. That is, $\{ \cdot, \cdot \}_\subM$ was computed
from the canonical Poisson bracket on $T^*Q$ but written in terms of
the adapted coordinates $(r^\a, s^a, \tilde p_\a, \tilde p_a)$. As a
result, the almost Poisson bracket $\{ \cdot, \cdot \}_\subM$ on
$\M$, written in local coordinates $(r^\a, s^a, \tilde p_\a)$, has
the following form \cite{MarsdenKoon}
\begin{equation}\label{Eq:NHbracket-coord} \{q^i,q^j\}_\subM  = 0, \quad
\{r^\a,\tilde p_\beta \}_\subM  = \delta_\a^\beta, \quad
\{s^a,\tilde p_\a\}_\subM  = -A_\a^a,\quad
\{\tilde p_\a,\tilde p_\beta\}_\subM  = K_{\a\beta}^b J_b^\gamma \tilde p_\gamma
\end{equation}

\section{The coordinate version of the Jacobiator formula}
\label{S:coord}

\subsection{Interpretation of the adapted coordinates}

In this section, we will relate the choice of the coordinates
proposed in \cite{MarsdenKoon} with the choice of a complement $\W$
done in \cite{paula} (see \eqref{Eq:CoordMK} and
\eqref{Eq:SplittingOfTM}, respectively).  We will also connect the
{\it curvature} \eqref{Eq:Kw-MK} with the 2-form \eqref{Def:K}, and
the nonholonomic bivector $\pi_\nh$ with the bracket $\{ \cdot ,
\cdot \}_\subM$ given in \eqref{Eq:Pinh} and
\eqref{Eq:NHbracket-coord}, respectively.

Consider a nonholonomic system on a manifold $Q$ given by a
lagrangian $L$ and a nonintegrable distribution $D$. Let us consider
local coordinates $(r^\a, s^a)$ adapted to the constraints as in
\eqref{Eq:CoordMK}.

\begin{lemma}\label{L:Equival-Coord}
The choice of coordinates $(r^\a,s^a)$ adapted to the constraints
\eqref{Eq:CoordMK}, induce a complement $W$ of $D$ on $TQ$ such that
\begin{equation} \label{Eq:Decomp-TQ}
TQ = D \oplus W, \qquad \mbox{where} \quad  W =
\textup{span}\left\{\frac{\partial}{\partial s^a} \right\}.
\end{equation}
\end{lemma}

\medskip

The projection  $P_W: TQ \to W$ associated to the decomposition
\eqref{Eq:Decomp-TQ}  is interpreted in \cite{MarsdenKoon} as the
Ehresmann connection $A$ \eqref{Eq:A}. In this context we compare
the curvature ${\bf K}_W$ defined in \eqref{Eq:Kw-MK} (see
\cite{MarsdenKoon}) with the $\W$-valued 2-form ${\bf K}_\subW$
defined in \eqref{Def:K}.

Recall that the submanifold $\M = \kappa^\flat(D) \subset T^*Q$ is
described by local coordinates $(r^\a, s^a; \tilde p_\a)$ (see
\eqref{Eq:M}). Locally $T^*\M$ is generated by the basis
$\mathfrak{B}_{T^*\M} = \{dr^\a,\epsilon^a, d\tilde p_\a\}$. During
the rest of the paper, when there is no risk of confusion, we will
use the same notation for 1-forms on $Q$ and their pull back to $\M$
and $T^*Q$, (i.e., $\tau^* dr^\a= dr^\a$ and $\tau^*\epsilon^a=
\epsilon^a$ where $\tau:\M \to Q$ is the canonical  projection).

Since $\tau$-projectable vector fields generate $T\M$ at each point, we can consider a complement $\W$ of $\C$ generated by $\tau$-projectable vector fields $Z_a$ such that $T\tau(Z_a) \in W$. 
That is, \begin{equation}\label{Eq:C+W-coord}
\C=\textup{span}\left\{ X_\a, \frac{\partial}{\partial \tilde p_\a}
\right\} \quad \mbox{and} \quad \W =\textup{span}\left\{Z_a \ : \ T\tau(Z_a) = \frac{\partial}{\partial s^a} \right\}.
\end{equation}

\begin{lemma} \label{L:KGamma} Let $\W$ be a complement of $\C$ as in \eqref{Eq:C+W-coord} where $W$ is the complement of $D$ induced by the coordinates $(r^\a, s^a)$ as in Lemma \ref{L:Equival-Coord}. 
\begin{enumerate}
\item[$(i)$]
The $\W$-valued 2-form ${\bf K}_{\subW}$ and the curvature ${\bf K}_W$, defined in \eqref{Def:K} and \eqref{Eq:Kw-MK} respectively, are related, at each  $X,Y \in T\M$, by ${\bf K}_W (T\tau(X),T\tau(Y)) = T\tau({\bf K}_{\subW}(X,Y)).$
In local coordinates $(r^\a, s^a)$ adapted to the constraints
\eqref{Eq:CoordMK}, the following holds:
$$
{\bf K}_{\subW} |_\C = (C_{\a\beta}^a dr^\a\wedge dr^\beta |_\C )\otimes Z_a.
$$
\item[$(ii)$]
Let $\bar{\W}$ be a different complement of $\C$ such that $T\tau(\W) = T\tau(\bar{\W}) = W$.  For $X,Y \in \Gamma(\C)$ we have
$$
{\bf K}_{\subW}(X,Y) - {\bf K}_{\bar{\subW}} (X,Y) \in \Gamma(\C).
$$
\end{enumerate}
\end{lemma}
\begin{proof}
$(i)$ During this proof and to avoid confusion, we will work with the
basis $\{\tau^*dr^\a,\tau^*\epsilon^a, d\tilde p_\a\}$ of $T^*\M$,
keeping $dr^\a$ and $ds^a$ to denote 1-forms on $Q$. Let us consider
the basis $\mathfrak{B}=\{X_\a, \frac{\partial}{\partial \tilde
p_\a}, Z_a\}$ of $T\M$ adapted to $\C \oplus \W$ and its dual $\mathfrak{B}^*= \{\tau^* dr^a, \Psi_\a , \tau^*\epsilon^a\}$ where $\Psi_\beta(X_\a) = \Psi_\beta(Z_a) = 0$ and
$\Psi_\beta( \frac{\partial}{\partial \tilde p_\a} ) =
\delta_{\a\beta}$.
Then, for $X,Y \in \Gamma(\C)$,
$$
{\bf K}_{\subW}(X,Y) =   - P_{\subW}( [X,Y] ) = - \tau^*\epsilon^a([X,Y])Z_a = d\tau^*\epsilon^a(X,Y) Z_a = d\epsilon^a(T\tau(X), T\tau(Y)) Z_a.
$$
Therefore, $T\tau({\bf K}_{\subW}(X,Y)) = d\epsilon^a(T\tau(X), T\tau(Y))  \otimes \frac{\partial}{\partial s^a} = {\bf K}_W(T\tau(X), T\tau(Y)).$
Finally, since $T\tau(X), T\tau(Y) \in \Gamma(D)$ (see \eqref{Eq:C})
and using \eqref{Eq:depsilon-coord} we obtain
$$
{\bf K}_{\subW}|_\C =   (C_{\a\beta}^a \tau^*dr^\a \wedge \tau^*dr^\beta |_\C ) \otimes Z_a.
$$
Using our simplified notation ($\tau^*dr^\a = dr^\a$) we obtain the desired formula.

$(ii)$ Let $\mathfrak{B}$ and $\mathfrak{B}^*$ be the basis as in item $(i)$. Consider also $\bar{\mathfrak{B}}=\{X_\a, \frac{\partial}{\partial \tilde
p_\a}, \bar{Z}_a\}$ a basis of $T\M$ adapted to $T\M=\C \oplus \bar{\W}$ such that $T\tau(\bar{Z}_a) = \frac{\partial}{\partial s^a}$ and its dual  $\bar{\mathfrak{B}}^*=\{dr^\a,\bar{\Psi}_\a, \epsilon^a\}$,
such that  $\bar{\Psi}_\beta(X_\a) = \bar{\Psi}_\beta(\bar{Z}_a) = 0$ and $\bar{\Psi}_\beta(\frac{\partial}{\partial \tilde p_\a} ) = \delta_{\a\beta}$. Then we have that, for $X,Y \in \C$,
$$
{\bf K}_{\bar{\subW}} (X,Y) = - P_{\bar{\subW}} ([X,Y]) = \epsilon^a([X,Y]) \bar{Z}_a  = {\bf K}_{\subW}(X,Y) + \epsilon^a([X,Y]) \otimes(\bar{Z}_a - Z_a).
$$
Since  $\bar{Z}_a - Z_a \in \textup{Ker}\, T\tau \subset \C$ then
${\bf K}_\subW(X,Y) - {\bf K}_{\bar{\subW}} (X,Y)  \in \C$.
\end{proof}

\begin{remark} \label{Prop:Ksemi-basic} Note that the coordinates description of 
${\bf K}_\subW$ shows that it is semi-basic with respect to the bundle projection $\tau:\M \to Q$, i.e., \ ${\bf i}_X {\bf K}_\subW =0$ \ if \ $T\tau(X)=0$. This is in agreement with \cite[Prop.~3.1]{paula} 
\end{remark}

In order to write the nonholonomic bivector $\pi_\nh$ using
\eqref{Eq:Pinh} but in local coordinates $(r^\a, s^a; \tilde p_\a)$
on $\M$ we study the local description of the 2-section
$\Omega_\subM|_\C$.

The canonical 1-form $\Theta_Q$ on $T^*Q$ is given, in local
coordinates $(r^\a, s^a; \tilde p_\a, \tilde p_a)$, by $\Theta_Q =
\tilde p_\a dr^\a + \tilde p_a \epsilon^\a$. Then, it is
straightforward to see that the canonical 2-form $\Omega_Q$ is
written locally as
$$
\Omega_Q= dr^\a \wedge d\tilde p_\a  +  \epsilon^a \wedge d\tilde p_a - \tilde p_a d\epsilon^a.
$$
Recall that $\iota: \M \to T^*Q$ is the natural inclusion, so the pull back of $\Omega_Q$ to $\M$ is given by
\begin{equation} \label{Eq:OmM-coord}
\Omega_\subM = \iota^*\Omega_Q = dr^\a \wedge d\tilde p_\a  + \iota^* \epsilon^a \wedge d\iota^*(\tilde p_a) - \iota^*(\tilde p_a) d(\iota^*\epsilon^a),
\end{equation}
where $dr^\a$ and $d\tilde p_\a$ are considered as 1-forms on $\M$.

Therefore,
\begin{equation}\label{Eq:OmC-coord}
\begin{split}
\Omega_\subM |_\C & = dr^\a \wedge d\tilde p_\a - \iota^*(\tilde p_a) \iota^*(d \epsilon^a) |_\C \\
& =  dr^\a \wedge d\tilde p_\a -  J_a^\delta \tilde p_\delta C_{\a\beta}^a dr^\a \wedge dr^\beta |_\C,
\end{split}
\end{equation}
where in the last equation we use \eqref{Eq:M} and the coordinate
version of $d \epsilon |_D$ given in \eqref{Eq:depsilon-coord}.
Applying \eqref{Eq:Pinh} to the 2-form $\Omega_\subM$ and $\C$,
given in \eqref{Eq:OmM-coord} and \eqref{Eq:C+W-coord} respectively, we
compute the nonholonomic bivector field $\pi_\nh$ on $\M$:
\begin{equation}\label{Eq:NH-bivector}
\pi_\nh^\sharp(dr^\a) =\frac{\partial}{\partial \tilde p_\a},\qquad
\pi_\nh^\sharp(ds^a) = -A_\a^a \frac{\partial}{\partial \tilde
p_\a}, \qquad \pi_\nh^\sharp(d\tilde p_\a) = - X_\a + J_a^\delta
\tilde p_\delta K_{\alpha\beta}^a\frac{\partial}{\partial \tilde
p_\beta}.
\end{equation}

\begin{lemma}\label{L:NHbracket}
The almost Poisson bracket $\{\cdot , \cdot\}_\subM$ given in
\eqref{Eq:NHbracket-coord} (see \cite[Theorem 2.1]{MarsdenKoon}) is
the coordinate version of the nonholonomic bracket $\{\cdot , \cdot
\}_{\emph\nh}$ associated to the bivector field $\pi_{\emph\nh}$
obtained from \eqref{Eq:Pinh}.
\end{lemma}

\subsection{The Jacobiator in adapted coordinates}
Consider a nonholonomic system on a manifold $Q$ given by a lagrangian $L$ and a constraint distribution $D$ such that
$\epsilon^a$,  for $a=1,...,k$, are 1-forms generating $D^\circ$.
Consider local coordinates $(r^\a,s^a)$ on $Q$ adapted to the
constraints as in \eqref{Eq:CoordMK}. Let $(r^\a, s^a;\tilde p_\a)$
be the coordinates on the manifold $\M=\kappa^\flat(D)$. By Lemma
\ref{L:NHbracket}, the almost Poisson bracket $\{\cdot, \cdot\}_\M$
\eqref{Eq:NHbracket-coord} is the coordinate version of the bivector
field $\pi_\nh$ given in \eqref{Eq:Pinh}, and thus Koon-Marsden
formula for the Jacobiator can be written directly with respect to
$\{ \cdot , \cdot \}_\nh$.

\begin{theorem}\cite[Sec.~2.5]{MarsdenKoon}
The Jacobiator of the nonholonomic bracket $\{\cdot, \cdot
\}_{\emph\nh}$, in coordinates $(r^a, s^a;\tilde p_\a)$ on $\M$,  is
given by the following formula
\begin{eqnarray}
\{\tilde p_\gamma, \{r^\a,\tilde p_\beta \}_{\emph\nh} \}_{\emph\nh}  + \textup{cyclic} & = & J_b^\a K_{\beta\gamma}^b, \nonumber\\
\{\tilde p_\beta, \{s^a,\tilde p_\a\}_{\emph\nh} \}_{\emph\nh}  + \textup{cyclic} & = & -K_{\a\beta}^a -     A_\gamma^a J_b^\gamma K_{\a\beta}^b, \label{Eq:JacobMK}\\
\{\tilde p_\gamma, \{\tilde p_\a,\tilde p_\beta\}_{\emph\nh} \}_{\emph\nh}  +  \textup{cyclic} & = & \tilde p_\tau J_a^\tau  \frac{\partial A_\gamma^a}{\partial s^b}K_{\a\beta}^b  + \tilde p_\tau J_a^\tau K_{\delta\gamma}^a J_b^\delta K_{\a\beta}^b - \tilde p_\tau K_{\a\beta}^b \left( \frac{\partial J_b^\tau}{\partial
r^\gamma} - A_\gamma^a \frac{\partial J_b^\tau}{\partial s^a} \right) +  \textup{cyclic}, \nonumber
\end{eqnarray}
with all other combinations equal to zero and where $J_b^\a$,
$K_{\a\beta}^a$ and $A_\a^a$ are the functions on $Q$ defined in
\eqref{Eq:M}, \eqref{Eq:Calphabeta} and \eqref{Eq:CoordMK}, respectively.
\end{theorem}

The next result relates the coordinate formula \eqref{Eq:JacobMK} of
the Jacobiator with the coordinate-free formula given in Theorem
\ref{T:MarsdenKoon}.

\begin{theorem} \label{T:Equivalence}
Let $(r^\a, s^a)$ be coordinates on $Q$ adapted to the constraints
as in \eqref{Eq:CoordMK} and let $W$ be the complement of $D$
induced by the coordinates (Lemma \ref{L:Equival-Coord}). The
Koon-Marsden Jacobiator formula \eqref{Eq:JacobMK} for the
nonholonomic bracket $\{\cdot, \cdot \}_{\emph\nh}$ is the
coordinate version of the Jacobiator formula given in Theorem
\ref{T:MarsdenKoon} for $\W$ any complement of $\C$ as in \eqref{Eq:C+W-coord}.
\end{theorem}

\begin{proof}
In order to prove the equivalence we write the Schouten bracket
$[\pi_\nh,\pi_\nh]$ using Theorem \ref{T:MarsdenKoon} evaluated on
the elements $\{dr^\a, ds^a, d\tilde p_\a\}$.

First, observe that by Remark \ref{Prop:Ksemi-basic}, the 2-form
${\bf K}_\subW$ defined in \eqref{Def:K} is annihilated by any of the elements $\pi_\nh^\sharp(dr^\a)$ or $\pi_\nh^\sharp(ds^\a)$ (see \eqref{Eq:NH-bivector}).  On the other hand, by Lemma \ref{L:KGamma}$(ii)$, we have that   ${\bf
K}_\subW(\pi_\nh^\sharp(d\tilde p_\a), \pi_\nh^\sharp(d\tilde
p_\beta)) = K_{\a\beta}^a Z_a$, where $Z_a \in T\M$ such that $T\tau(Z_a)= \frac{\partial}{\partial s^a}$. Moreover, observe that $\epsilon^a(Z_b) =
\delta_{\a\beta}$ and $dr^\a(Z_a)=0$.

Therefore, using the coordinate version of $\Omega_\M$
\eqref{Eq:OmM-coord} in Theorem \ref{T:MarsdenKoon} we obtain
\begin{equation*} \begin{split}
\frac{1}{2}[\pi_\nh, \pi_\nh]  (dr^\a,d\tilde p_\beta,d\tilde p_\gamma)
&= \Omega_\subM ({\bf K}_\subW(\pi_\nh^\sharp(d\tilde p_\beta), \pi_\nh^\sharp(d\tilde p_\gamma) ), \pi_\nh^\sharp(d r^\a) ) - dr^\a ( {\bf K}_\subW(\pi_\nh^\sharp(d\tilde p_\beta),\pi_\nh^\sharp(d\tilde p_\gamma) ) ) \\
& = \ J_a^\a K_{\beta\gamma}^a, \\ 
\frac{1}{2} [\pi_\nh, \pi_\nh] (d s^a, d\tilde p_\alpha,d\tilde p_\beta) & =
 \Omega_\subM({\bf K}_\subW(\pi_\nh^\sharp(d\tilde p_\a),\pi_\nh^\sharp(d\tilde p_\beta) ), \pi_\nh^\sharp(d s^a) ) - ds^a ( {\bf K}_\subW(\pi_\nh^\sharp(d\tilde p_\a),\pi_\nh^\sharp(d\tilde p_\beta) ) ) \\ & = \ - K_{\a\beta}^b A_\gamma^a J_b^\gamma  - K_{\a\beta}^a.
\end{split}
\end{equation*}
Finally, let $Y_a := Z_a-\frac{\partial}{\partial s^a} \in
\textup{Ker}\, T\tau \subset \C$. Then, we have that
\begin{equation*} \begin{split}
\frac{1}{2}[\pi_\nh, \pi_\nh] (d \tilde p_\a, d\tilde p_\beta, d \tilde
p_\gamma)
& = \Omega_\subM(K_{\a\beta}^a Z_a, \pi_\nh^\sharp(d \tilde p_\gamma) ) - d\tilde p_\gamma(K_{\a\beta}^a Z_a )  + \textup{cyclic} \\
& = \Omega_\subM(K_{\a\beta}^a \frac{\partial}{\partial s^a}, \pi_\nh^\sharp(d \tilde p_\gamma) ) + \Omega_\subM(K_{\a\beta}^a Y_a, \pi_\nh^\sharp(d \tilde p_\gamma) )  -  d\tilde p_\gamma(K_{\a\beta}^a Y_a )  + \textup{cyclic} \\
& = \Omega_\subM(K_{\a\beta}^a \frac{\partial}{\partial s^a}, \pi_\nh^\sharp(d \tilde p_\gamma) ) + \textup{cyclic} \\
& = \  \tilde p_\tau J_a^\tau  \frac{\partial A_\gamma^a}{\partial s^b}K_{\a\beta}^b  + \tilde p_\tau J_a^\tau K_{\delta\gamma}^a J_b^\delta K_{\a\beta}^b - \tilde p_\tau K_{\a\beta}^b \left( \frac{\partial J_b^\tau}{\partial
r^\gamma} - A_\gamma^a \frac{\partial J_b^\tau}{\partial s^a} \right) +  \textup{cyclic}.
\end{split} \end{equation*}
The Jacobiator on the other combinations of elements of the basis
$\{dr^\a, ds^a,d\tilde p_\a\}$ is zero. Thus,  the relation
\eqref{E:Jacobi} implies that the Jacobiator formula in Theorem
\ref{T:MarsdenKoon} evaluated in coordinates \eqref{Eq:CoordMK}
gives the Koon-Marsden formula \eqref{Eq:JacobMK}.

Observe that in this proof we are implicitly using Lemma
\ref{L:KGamma} for $\W$ and $\W_0 = \textup{span}\left\{
\frac{\partial}{\partial s^a} \right\}$.

\end{proof}

\begin{remark}
From \eqref{Eq:Jacobiator} it is straightforward to see that if the
2-form ${\bf K}_\subW$ is zero then the bivector $\pi_\nh$ is
Poisson. On the other hand, it was observed in \cite{MarsdenKoon}
that if the curvature ${\bf K}_W$ is zero then the Jacobi identity
of $\{ \cdot , \cdot \}_\nh$ is satisfied.  Using the equivalence
between ${\bf K}_\subW$ and ${\bf K}_W$ (Lemma
\ref{L:KGamma}$(i)$) we see that both 2-forms are zero when $D$
is involutive, i.e., the system is holonomic.
\end{remark}

\begin{remark}(Symmetries)
If the nonholonomic system admits a group of symmetries $G$
then $\pi_\nh$ is $G$-invariant with respect to the induced (lifted)
action on $\M$. As a consequence, the orbit projection $\M \to \M/G$
induces a reduced bivector field $\pi_\red^\nh$ on $\M/G$ describing
the reduced dynamics. Consider $(r^\a, s^a)$ adapted coordinates to
the constraints as in \eqref{Eq:CoordMK} and $W$ the induced
complement of $D$ in $TQ$ given by Lemma \ref{L:Equival-Coord}.  Let
$V$ (respectively $\V$) be the tangent space to the orbit of the
$G$-action on $Q$ (respect. on $\M$). If $W \subset V$ then there is a unique choice of the complement $\W$ contained in $\V$:
$$
\W:= (T\tau |_\V)^{-1}(W).
$$
With this choice of $\W$, Theorem \ref{T:MarsdenKoon} induces a formula for the
Jacobiator of the reduced bivector $\pi_\red^\nh$ (see
\cite[Sec.4]{paula}).

There are a number of examples of systems verifying  that the complement $W$
induced by the coordinates adapted to the constraints
\eqref{Eq:CoordMK}  (as in Lemma \ref{L:Equival-Coord}) is vertical with respect to a $G$-symmetry, including the vertical rolling disk, the nonholonomic particle and the Chaplygin sphere, see \cite[Sec.~7]{paula}.

\end{remark}

On the other hand, it may happen that a given example is described
in coordinates that are not adapted to the constraints. Then, it is
better to use the coordinate free formula of Theorem
\ref{T:MarsdenKoon}.

\section{Example: the snakeboard}\label{S:ex}

The snakeboard describes the dynamics of a board with two sets of
actuated wheels, one on each end of the board.  A human rider
generates forward motion by twisting his body back and forth, and
thus producing a movement on the wheels. This effect is modeled as a
momentum wheel which sits in the middle of the board and is allowed
to spin about the vertical axis. The configuration of the snakeboard
is given by the position and orientation of the board in the plane,
the angle of the momentum wheel and the angles of the back and front
wheels. Therefore, the configuration manifold $Q$ is given by
$Q=SE(2)\times (-\pi/2,\pi/2) \times S^1$ with local coordinates
$q=(x,y,\theta,\psi,\phi)$, where $(x,y,\theta)$ represents the
position and orientation of the center of the board, $\psi$ is the
angle of the momentum wheel relative to the board and $\phi$ is the
angle of the front and back wheel as in \cite{Ostrowski} (for
details see \cite{BKMM} and \cite{MarsdenKoon}).

The Lagrangian is given by
$$
L(q, \dot q) = \frac{m}{2} (\dot x^2+\dot y^2) + \frac{mr^2}{2} \dot
\theta^2 + \frac{J_0}{2}\dot \psi^2 + J_0\dot \psi \dot \theta + J_1
\dot \phi^2,$$ where $m$ the total mass of the board, $r$ is the
distance between the center of the board and the wheels, $J_0$ is
the inertia of the rotor and $J_1$ is the inertia of each wheel.

The (nonintegrable) constraint distribution $D$ is given by the
annihilator of the following 1-forms:
\begin{equation} \label{Ex:1forms-const}
\begin{split}
\epsilon^1&=-\sin(\theta + \phi)dx+\cos(\theta+\phi)dy-r\cos\phi d\theta\\
\epsilon^2&=-\sin(\theta - \phi)dx+\cos(\theta-\phi)dy+r\cos\phi d\theta.
\end{split}
\end{equation}

\begin{remark}
The coordinates $(x,y,\theta,\psi,\phi)$ on $Q$ are not adapted to
the 1-forms of constraints $\epsilon^1, \epsilon^2$. In \cite{MK2} a
simplified version is considered where, taking $\phi \neq 0$, it is
possible to write the 1-forms of constraints in such a way that
$(x,y,\theta,\psi,\phi)$ are adapted coordinates as in
\eqref{Eq:CoordMK}. In this paper, we will work with the 1-forms
given in \eqref{Ex:1forms-const}, so, our coordinates in $Q$ are not
adapted to the constraints, even though these are the coordinates
chosen in \cite{MarsdenKoon} to study the reduction by the group of
symmetries $SE(2)$.
\end{remark}

The distribution $D$ on $Q$ is given by
\begin{equation*}
D=\textup{span} \left\{X_\psi:= \frac{\partial}{\partial \psi}, \
X_\phi:= \frac{\partial}{\partial \phi}, \  X_\subS:= -2r \cos^2\phi
\cos\theta \frac{\partial}{\partial x} -2r \cos^2\phi\sin\theta
\frac{\partial}{\partial y} +\sin(2\phi) \frac{\partial}{\partial
\theta} \right\}.
\end{equation*}

We choose the complement $W$ of $D$ generated by $\{X_1,X_2\}$ so
that $\epsilon^a(X_b) =  \delta^a_ b$ for $a,b=1,2$, that is
\begin{equation*}
\begin{split}
W=\textup{span} \left\{\begin{array}{c} \, \\ \, \end{array}\right. & X_1:= -\frac{1}{2}\sin\theta \sec\phi \frac{\partial}{\partial x} +\frac{1}{2} \cos\theta\sec\phi \frac{\partial}{\partial y}-\frac{1}{2r}\sec\phi \frac{\partial}{\partial \theta},  \\
& X_2:=\left. -\frac{1}{2}\sin\theta \sec\phi \frac{\partial}{\partial x} +\frac{1}{2} \cos\theta\sec\phi \frac{\partial}{\partial y}+\frac{1}{2r}\sec\phi \frac{\partial}{\partial \theta} \   \right\}.
\end{split}
\end{equation*}
Consider the dual basis $\mathfrak{B}_{TQ} = \{
X_\psi, X_\phi,X_\subS, X_1, X_2\}$ and $\mathfrak{B}_{T^*Q}= \{ d\psi, d\phi,
\alpha_\subS, \epsilon^1, \epsilon^2\}$ where
$$\a_\subS= -\frac{1}{2r} \cos\theta \sec^2\phi dx-\frac{1}{2r} \sin\theta \sec^2\phi dy.$$
Let us denote by $(q;v_\psi, v_\phi, v_\subS, v_1, v_2)$ the
coordinates on $TQ$ associated with the basis $\mathfrak{B}_{TQ}$
while $(q; \tilde p_\psi, \tilde p_\phi, \tilde p_\subS, \tilde p_1,
\tilde p_2)$ denote the coordinates on $T^*Q$ associated to
$\mathfrak{B}_{T^*Q}$.

The submanifold $\M = \kappa^\flat(D)= \textup{span}\{
\kappa^\flat(X_\psi),\kappa^\flat(X_\phi),\kappa^\flat(X_\subS) \}$
is defined in coordinates by
\begin{equation} \label{Ex:M}
\M = \{(q; \tilde p_\psi, \tilde p_\phi, \tilde p_\subS, \tilde p_1, \tilde p_2) \ : \ \tilde p_1= -\tilde p_2 = J_1(\phi)\tilde p_\subS +J_2(\phi)\tilde p_\psi \},
\end{equation}
where
$$J_1(\phi)= \frac{mr}{4(r^2m-J_0\sin^2\phi)} \sin\phi \sec^2\phi \qquad  \mbox{and} \qquad J_2(\phi) =- J_1(\phi)\sin(2\phi) .
$$

In order to compute the nonholonomic bivector $\pi_\nh$ describing
the dynamics, we  write the 2-form
$\Omega_\subM$ and the 2-section $\Omega_\M|_\C$ in our local coordinates. The canonical 1-form $\Theta_Q$ on $T^*Q$ is given by $\Theta_Q = \tilde p_\psi d\psi + \tilde p_\phi d\phi + \tilde
p_\subS \alpha_\subS + \tilde p_a \epsilon^a$. Then,
$$
\Omega_Q= d\psi \wedge d\tilde p_\psi + d\phi \wedge d\tilde p_\phi
+ \alpha_\subS \wedge d\tilde p_\subS - \tilde p_\subS d\alpha_\subS
+\epsilon^1 \wedge d\tilde p_1+ \epsilon^2 \wedge d\tilde p_2 -
\tilde p_1d\epsilon^1 -\tilde p_2 d\epsilon^2,
$$

Let us consider the basis
$\mathfrak{B}_{T^*\M}=\{d\phi,d\psi,\alpha_\subS,\epsilon^1,\epsilon^2,d\tilde
p_\phi,d\tilde p_\psi, d\tilde p_\subS\}$ of $T^*\M$ (here we are
using the same notation for the pullbacks of the forms to $\M$).
Recall that, on $\M$, $\tilde p_1$ and $\tilde p_2$ are given by
\eqref{Ex:M} and denoting $J_i = J_i(\phi)$ for $i=1,2$ we obtain
\begin{equation} \label{Ex:Om_M}
\begin{split}
\Omega_\subM= & d\psi \wedge d\tilde p_\psi + d\phi \wedge d\tilde p_\phi + \alpha_\subS \wedge d\tilde p_\subS - \tilde p_\subS d\alpha_\subS +
J_1\epsilon^1 \wedge d\tilde p_\subS + J_2 \epsilon^1 \wedge d \tilde p_\psi + \tilde p_\subS J'_1 \epsilon^1\wedge d\phi + \tilde p_\psi J'_2 \epsilon^1\wedge d\phi \\
&  - J_1\epsilon^2 \wedge d\tilde p_\subS - J_2 \epsilon^2 \wedge d \tilde p_\psi - \tilde p_\subS J'_1 \epsilon^2\wedge d\phi - \tilde p_\psi J'_2 \epsilon^2\wedge d\phi - (J_1\tilde p_\subS +J_2\tilde p_\psi)(d\epsilon^1-d\epsilon^2).
\end{split}
\end{equation}

On $T\M$ consider the dual basis $\mathfrak{B}_{T\M} = \left\{
X_\psi \, , \, X_\phi \, , \,  X_\subS \, , \, X_1\, , \, X_2\, , \,
\frac{\partial}{\partial \tilde p_\psi} \, , \,
\frac{\partial}{\partial \tilde p_\phi} \, , \,
\frac{\partial}{\partial p_\subS} \right\}$  associated to
$\mathfrak{B}_{T^*\M}$. Therefore, we can decompose $T\M = \C \oplus
\W$ such that
\begin{equation} \label{Ex:C+W}
\C= \textup{span}\left\{ X_\psi \, , \, X_\phi \, , \,  X_\subS \, , \, \frac{\partial}{\partial \tilde p_\psi} \, , \, \frac{\partial}{\partial \tilde p_\phi} \, , \, \frac{\partial}{\partial p_\subS}\right\} \qquad \W =\textup{span} \left\{X_1 \, , \, X_2 \right\}.
\end{equation}

Therefore, using that $d\epsilon^a |_\C = (-1)^a 2r \cos\phi
\alpha_\subS \wedge d\phi |_\C$ for $a=1,2$ and that $d\alpha_\subS
|_\C = 2\tan\phi d\phi \wedge \alpha_\subS|_\C$, the 2-section
$\Omega_\subM|_\C$ is given by
$$
\Omega_\M|_\C =  d\psi \wedge d\tilde p_\psi + d\phi \wedge d\tilde
p_\phi + \alpha_\subS \wedge d\tilde p_\subS - \tilde p_\subS
2\tan\phi d\phi \wedge \alpha_\subS + (J_1\tilde p_\subS +J_2\tilde
p_\psi) 4r \cos\phi  \alpha_\subS \wedge d\phi \ |_\C.
$$

Now, we compute the nonholonomic bracket $\pi_\nh$ using
\eqref{Eq:Pinh}
\begin{equation} \label{Ex:NHbracket}
\pi_\nh=\frac{\partial}{\partial \psi} \wedge \frac{\partial}{\partial \tilde p_\psi} + \frac{\partial}{\partial \phi} \wedge \frac{\partial}{\partial \tilde p_\phi} + X_\subS \wedge \frac{\partial}{\partial \tilde p_\subS} - (\tilde p_\subS 2\tan\phi + 4r (J_1\tilde p_\subS +J_2\tilde p_\psi)\cos\phi  )  \frac{\partial}{\partial \tilde p_\subS} \wedge \frac{\partial}{\partial \tilde p_\phi} .
\end{equation}

Therefore, the hamiltonian vector fields are
\begin{equation} \label{Ex:HamVectFields}
\begin{split}
\pi^\sharp_\nh(d\psi) & = \frac{\partial}{\partial \tilde p_\psi}, \qquad \pi_\nh^\sharp(d\phi)= \frac{\partial}{\partial \tilde p_\phi}, \\
\pi_\nh^\sharp(\alpha_\subS) & = \frac{\partial}{\partial \tilde p_\subS},  \qquad \pi_\nh^\sharp(\epsilon^i)=0, \qquad \pi_\nh^\sharp(d\tilde p_\psi) = - \frac{\partial}{\partial \psi},\\
\pi_\nh^\sharp(d\tilde p_\phi) &= - \frac{\partial}{\partial \phi} + (2\tan\phi \tilde p_\subS +4r\cos\phi(J_1\tilde p_\subS + J_2\tilde p_\psi) )\frac{\partial}{\partial \tilde p_\subS},\\
 \pi_\nh^\sharp(d\tilde p_\subS) &= - X_\subS - (2\tan\phi \tilde p_\subS +4r\cos\phi(J_1\tilde p_\subS + J_2\tilde p_\psi) )\frac{\partial}{\partial \tilde p_\phi}.
\end{split}
\end{equation}

In order to apply Theorem \ref{T:MarsdenKoon} to compute the
Jacobiator of $\pi_\nh$ we study the $\W$-valued 2-form ${\bf
K}_\subW$ defined in \eqref{Def:K} for $\W$ in \eqref{Ex:C+W}. For
$X,Y \in \C$ and using the dual basis $\mathfrak{B}_{T\M}$ and
$\mathfrak{B}_{T^*\M}$  we have that
\begin{equation*}
\begin{split}
{\bf K}_\subW (X,Y) &  =  -P_\subW([X,Y]) = -\epsilon^1([X,Y]) X_1 - \epsilon^2([X,Y]) X_2 \\
& = d\epsilon^1(X,Y) X_1 + d\epsilon^2(X,Y) X_2.
\end{split}
\end{equation*}
Therefore,
\begin{equation} \label{Ex:Kw}
{\bf K}_\subW |_\C = -2r \cos(\phi) (\alpha_\subS \wedge d\phi \, |_\C) \otimes (X_1 - X_2).
\end{equation}

Finally, we consider the 2-forms $\Omega_\subM$ and ${\bf K}_\subW$, described in \eqref{Ex:Om_M} and \eqref{Ex:Kw} and the vector fields
\eqref{Ex:HamVectFields}, to obtain, by \eqref{Eq:Jacobiator},  that
\begin{equation} \label{Ex:JacobEx}
\begin{split}
& [\pi_\nh, \pi_\nh] (d\tilde p_\phi, d\tilde p_\subS,d\psi) = 4r\cos(\phi) J_2(\phi), \\
&[\pi_\nh, \pi_\nh] (d\tilde p_\phi, d\tilde p_\subS,\alpha)  = 4r\cos(\phi) J_1(\phi)\\
& [\pi_\nh, \pi_\nh] (d\tilde p_\phi, d\tilde p_\subS,\epsilon^i)  =
(-1)^i 2r\cos(\phi), \;\;\; i=1,2,
\end{split}
\end{equation}
while on other combination of elements the Jacobiator is zero.

This example admits a symmetry given by the Lie group $SE(2)$, see
\cite{MarsdenKoon}. The reduced manifold $\M/G$ is $S^1\times
S(-\pi/2,-\pi/2) \times \R^3$ and the nonholonomic bivector field
$\pi_\nh$ is invariant by the orbit projection $\rho:\M \to \M/G$.
Thus, on $\M/G$ we have the reduced nonholonomic bivector defined at
each $\alpha \in T^*(\M/G)$ by
$$
(\pi_\red^\nh)^\sharp(\alpha) = T\!\rho \, \pi_\nh^\sharp (\rho^*\alpha).
$$
The Jacobiator of the reduced nonholonomic bivector field
$\pi_\red^\nh$ satisfies
$$
[\pi_\red^\nh, \pi_\red^\nh] (\alpha, \beta,\gamma) = T\rho \left(
[\pi_\nh, \pi_\nh]  (\rho^*\alpha, \rho^*\beta,\rho^*\gamma)
\right)
$$
for $\alpha, \beta, \gamma \in T^*(\M/G)$. So, in our example it is
simple to compute the Jacobiator of $\pi_\red^\nh$. Taking into
account that, in local coordinates, the orbit projection $\rho : \M
\to \M/G$ is given by $\rho(\psi,\phi,\theta,x,y, \tilde p_\psi,
\tilde p_\phi,\tilde p_\subS)= (\psi,\phi,\tilde p_\psi, \tilde
p_\phi,\tilde p_\subS)$, the Jacobiator of the reduced bivector
$\pi_\red^\nh$ describing the dynamics is given by
$$
[\pi_\red^\nh, \pi_\red^\nh] (d\tilde p_\phi, d\tilde p_\subS,d\psi) = 4r\cos(\phi) J_2(\phi)
$$
while on other elements of $T^*(\M/G)$ is zero.

Just to complete the example we can write, in our coordinates, the
reduced bivector field $\pi_\red^\nh$ on $\M/G$:
$$
\pi_\nh^\red = \frac{\partial}{\partial \psi} \wedge
\frac{\partial}{\partial \tilde p_\psi} + \frac{\partial}{\partial
\phi} \wedge \frac{\partial}{\partial \tilde p_\phi} - (\tilde
p_\subS 2\tan\phi + 4r (J_1\tilde p_\subS + J_2\tilde
p_\psi)\cos(\phi)  )  \frac{\partial}{\partial \tilde p_\subS}
\wedge \frac{\partial}{\partial \tilde p_\phi} .
$$


\begin{thebibliography}{99}

\bibitem{paula} P. Balseiro, The Jacobiator of nonholonomic systems and the geometry of reduced nonholonomic brackets. Preprint, 2013.

\bibitem{BS93} L. Bates and J. Sniatycki, Nonholonomic reduction,  \emph{ Rep. Math. Phys.} \textbf{32}, (1993),99--115.

\bibitem{BlochBook} A.~M. Bloch, {\em Non-holonomic mechanics and control}.
Springer Verlag, New York, (2003).

\bibitem{BKMM} A.~M. Bloch, P.~S.Krishnapasad , J.~E. Marsden
and R.~M. Murray,  Nonholonomic mechanical systems with symmetry.
{\em Arch. Rat. Mech. An.}, {\bf 136}, (1996), 21--99.

\bibitem{BorisovMamaev2008}
 A.~V. Borisov and  I.~S. Mamaev, Conservation laws, hierarchy of dynamics and explicit
 integration of nonholonomic systems.
 \emph{Regul. Chaotic Dyn.}, {\bf 13}, (2008), 443--490.


\bibitem{Chapligyn_reducing_multiplier} S.~A. Chaplygin,  On the theory of the motion of nonholonomic systems. The reducing-multiplier theorem.
Translated from \emph{ Matematicheski\u{i}  sbornik} (Russian)  {\bf 28} (1911) , no. 1 by A. V. Getling.  \emph{ Regul.  Chaotic Dyn.}, {\bf 13},
(2008), 369--376.


\bibitem{CushmannBook} R.H.Cushman, H.Duistermaat, J.\v{S}niatycki, {\em Geometry of Nonholonomically Constrained Systems}, World Scientific Co. Pte. Ltd, (2010).

\bibitem{Fernandez} O. Fernandez, T. Mestdag and A. Bloch, A generalization of Chaplygin's reducibility
Theorem. \emph{Regul. Chaotic Dyn.} {\bf 14}, (2009), 635--655.

\bibitem{FedorovJovan} 
Yu.~N. Fedorov and  B. Jovanovi{\'c}, Nonholonomic LR systems as
generalized Chaplygin systems with an invariant measure and flows on
homogeneous spaces. \emph{J. Nonlinear Sci.}, \textbf{14} (2004),
341--381.

\bibitem{Naranjo2008} L. C. Garc\'ia-Naranjo, Reduction of  almost
Poisson brackets and Hamiltonization of the Chaplygin sphere.
{\em Disc. and Cont. Dyn. Syst. Series S}, {\bf 3}, (2010), 
37--60.


\bibitem{IbLeMaMa1999} A. Ibort, M. de Le\'on, J.~C. Marrero and D. Mart\'in de Diego,   Dirac
brackets in constrained dynamics. {\em Fortschr. Phys.} {\bf 47}, (1999), 
459--492.

\bibitem{JovaChap}  B. Jovanovi{\'c}, Hamiltonization and integrability of the Chaplygin sphere in $\R^n$.  \emph{J. Nonlinear Sci.}
 {\bf 20},  (2010),  
 569-593.

\bibitem{MovingFrames} J. Koiller, P. P. M. Rios, K. M. Ehlers,
Moving Frames for Cotangent Bundles. {\it Reports on Mathematical
Physics} No.49 (2), (2002), 225-238.



\bibitem{Marle1998}
Ch.~M. Marle,   Various approaches
to conservative and nonconservative nonholonomic systems, {\em Rep. Math.
Phys.}, {\bf 42},  (1998), 211--229.

\bibitem{MK2} W.S. Koon, J.E. Marsden, The Hamiltonian and Lagrangian Approaches to the Dynamics of Nonholonomic Systems, {\it Reports on Mathematical Physics}, {\bf 40}, (1997), 21-62.

\bibitem{MarsdenKoon} W.S.Koon, J.E.Marsden, The Poisson Reduction of Nonholonomic Mechanical Systems,
{\it Reports on Mathematical Physics}, {\bf 42}, (1998) 101-134.

\bibitem{MarsdenRatiu}  J.E. Marsden  and  T.~S.   Ratiu,
{\em Introduction to Mechanics and Symmetry. A basic exposition of
classical mechanical systems.} Second edition. Texts in Applied
Mathematics, {\bf 17}.  Springer-Verlag, New York, (1999).


\bibitem{Ostrowski} J. Ostrowski, {\it Geometric Perspectives on the Mechanics and Control of Undulatory Locomotion.} Ph.D. Thesis, California Institute of Technology. (1995).

\bibitem{SchaftMaschke1994} A. J. van der Schaft and  B. M. Maschke,  On the Hamiltonian
formulation of nonholonomic mechanical systems, {\it Rep. on Math. Phys.}
{\bf 34} (1994), 225--233.

\end{thebibliography}
\end{document}